\documentclass[12pt]{article}
\usepackage[utf8]{inputenc}
\usepackage[english]{babel}
\usepackage{graphicx}
\usepackage{amsmath}
\usepackage{amsthm}
\usepackage{amsmath,amssymb} 
\usepackage{subcaption}

\newcommand{\Bool}{\{0,1\}}

\theoremstyle{plain}
\newtheorem{theorem}{Theorem}
\newtheorem{lemma}{Lemma}
\newtheorem{definition}{Definition}

\theoremstyle{definition}

\begin{document}

\title{Second order conservative languages with a Maltsev polymorphism also have a majority polymorphism.}
\author{Vodolazskiy E.V.}
\date{}

\maketitle
\begin{abstract} The paper proves that for any second order conservative constraint language with a Maltsev polymorphism there is a majority polymorphism.
Moreover, the majority polymorphism can be defined by the Maltsev polymorphism.
\end{abstract}

\section{Introduction}

Polymorphisms have been known for a while as a powerful tool to describe tractable languages in constraint satisfaction problems (CSP) \cite{bulatov}.
It is known that a constraint language with a majority polymorphism results in a CSP that is solvable in polynomial time \cite{majority}. 
The same is true for the Maltsev polymorphism \cite{maltsev}. 
Even though both these classes are tractable, the algorithm for solving majority CSPs is much simpler and more efficient than the one that solves Maltsev CSPs.
We prove that any second order constraint language (the ones that only has constraints on one or two variables) with a Maltsev polymorphism also has a majority polymorphism.
Which makes it possible to use simpler algorithms designed for majority CSPs.

This result is closely related to the fact that Maltsev digraphs also have a majority polymorphism \cite{malDi} 
with the difference that we define an easy way to obtain a majority polymorphism given a Maltsev polymorphism. 
Our proof is also simpler and is presented in terms of constraint languages.

\section{Main definitions}

\begin{definition}
    For any finite set $X$ a second order language on this set is a subset of binary and unary functions on $X$.\\
	In other words, $\Gamma \subseteq \Bool^X \bigcup \Bool^{X{\times}X}$ is a second order language.
\end{definition}

\begin{definition}\label{conservativeLanguage}
	A language $\Gamma$ is called conservative if it contains all possible unary functions $f:X \rightarrow \Bool$:
	$$\forall f:X \rightarrow \Bool \Rightarrow f \in \Gamma.$$
\end{definition}

\begin{definition}\label{conservativeOperator}
	Operator $p:X \times X \times X \rightarrow X$ is called conservative if it's value is always equal to one of the arguments:
	$$\forall x, y, z \in X \Rightarrow p(x, y, z) \in \{x, y, z\}.$$
\end{definition}

\begin{definition}
	Operator $p:X \times X \times X \rightarrow X$ is called Maltsev if
	$$\forall x, y \in X \Rightarrow p(x, x, y) = p(y, x, x) = y.$$
\end{definition}

\begin{definition}
	Operator $p:X \times X \times X \rightarrow X$ is called a majority operator if
	$$\forall x, y \in X \Rightarrow p(x, x, y) = p(x, y, x) = p(y, x, x) = x.$$
\end{definition}

\begin{definition}
	Operator $p:X \times X \times X \rightarrow X$ is called a polymorphism of a $k$-ary function $X^k \rightarrow \Bool$ if
	$$\forall \overline{x}, \overline{y}, \overline{z} \in X^k: f(\overline{x}){=}f(\overline{y}){=}f(\overline{z}){=}1 \Rightarrow f(\overline{t}) = 1, \text{where } t_i = p(x_i, y_i, z_i), i=\overline{1,k}$$
\end{definition}

\begin{definition}
	Operator $p$ is called a polymorphism of a language $\Gamma$ if $p$ is a polymorphism of all functions $f \in \Gamma$ of the language.
\end{definition}

\section{The result}

\begin{lemma}\label{conservativePlymorphismLemma}
	Let $p:X \times X \times X \rightarrow X$ be a polymorphism of some conservative language~$\Gamma$. Then $p$ is a conservative operator.
\end{lemma}
\begin{proof}
    Assume the opposite, let $p$ be not conservative. Then, according to Definition \ref{conservativeOperator}
    there exists such a triple $x, y, z \in X$ that $p(x, y, z) \notin \{x, y, z\}$. 
    Consider such a function $f:X \rightarrow \Bool$ that
	$$f(t) = \begin{cases} 1, & \text{for } t=x,\\1, & \text{for } t=y,\\1, & \text{for } t=z,\\0, & \text{for } t \neq x, y, z.\\  \end{cases}$$
    Operator $p$ is not a polymorphism of $f$, because $f(x){=}f(y){=}f(z){=}1$ and $f(p(x, y, z)){=}0$, since $p(x, y, z) \neq x, y, z$.
    However, $f \in \Gamma$, as a function of one variable according to Definition \ref{conservativeLanguage}. 
	Therefore, $p$ is not a polymorphism of $\Gamma$, which contradicts our initial assumption.
\end{proof}

\begin{lemma}\label{conservativeIsPolymorphismLemma}
	Let $p:X \times X \times X \rightarrow X$ be a conservative operator and $f:X \rightarrow \Bool$ be a unary function. 
	Then $p$ is a polymorphism of $f$.
\end{lemma}
\begin{proof}
    Indeed, for any triple $x, y, z \in X$ that fulfills $f(x) = f(y) = f(z) = 1$ the relation $f(p(x, y, z)) = 1$ is also true, since $p(x, y, z) \in \{x, y, z\}$.
\end{proof}

\begin{definition}
	For each operator $p:X \times X \times X \rightarrow X$ define an operator $p':X \times X \times X \rightarrow X$ that we will call a derivative operator of $p$:
	$$p'(x, y, z) = \begin{cases}z, & \text{if } p(x, y, z) = x,\\ x, & \text{if } p(x, y, z) \neq x.\end{cases}$$
\end{definition}

\begin{lemma}\label{derivativeIsMajorityLemma}
	Let $p:X \times X \times X \rightarrow X$ be a Maltsev operator. Then its derivative operator $p'$ is a majority operator.
\end{lemma}
\begin{proof}
    Indeed, for any $x, y \in X$ the following is valid:
	\begin{itemize}
		\item $p'(x, x, y) = x$, since $p(x, x, y) = y$,
		\item $p'(x, y, x) = x$, since the derivative's operator value is always either the first or the last argument,
		\item $p'(x, y, y) = y$, since $p(x, y, y) = x$.
	\end{itemize}
\end{proof}

\begin{lemma}\label{derivativeIsPolymorphismLemma}
	Let $p:X \times X \times X \rightarrow X$ be a conservative Maltsev polymorphism of some binary function $f:X \times X \rightarrow \Bool$.
    Then its derivative operator $p'$ is also a polymorphism of $f$.
\end{lemma}
\newcommand{\byProposition}{by assumption}
\begin{proof}
    Assume that the function $f$ returns $1$ on some triple of arguments:
	$$f(x_1, x_2) = f(y_1, y_2) = f(z_1, z_2) = 1.$$
    We will show that by applying a derivative operator $p'$ to this triple we obtain a value on which the function $f$ is also equal to $1$.
	In other words, we will show that $f(p'(x_1, y_1, z_1), p'(x_2, y_2, z_2)) = 1$.
	
	Since operator $p$ is conservative, 
    the result of applying it to the triple of pairs of arguments can only be one of the following nine pairs: $x_1x_2$, $x_1y_2$, $x_1z_2$, $y_1x_2$, $y_1y_2$, $y_1z_2$, $z_1x_2$, $z_1y_2$, $z_1z_2$.
    For each of these nine cases we will show that the derivative operator $p'$ also gives a pair on which the function $f$ equals to $1$.\\
    The first column of the following table is the value obtained by applying $p$, the second column is the result of $p'$ and the third column is the proof given for this case.\\
	\begin{tabular}{l || l}
	\begin{tabular}{ c | c | l }
		$p$ & $p'$ & \\
		\hline
		$x_1x_2$ & $z_1z_2$ & \byProposition \\
		& & \\
		$x_1y_2$ & $z_1x_2$ & $\begin{array}{c}x_1y_2\\y_1y_2\\z_1z_2\\ \hline x_1z_2\end{array} \Rightarrow 
						       \begin{array}{c}x_1x_2\\x_1z_2\\z_1z_2\\ \hline z_1x_2\end{array}$ \\
		& & \\
		$x_1z_2$ & $z_1x_2$ & $\begin{array}{c}x_1x_2\\x_1z_2\\z_1z_2\\ \hline z_1x_2\end{array}$ \\
		& & \\
		$y_1x_2$ & $x_1z_2$ & $\begin{array}{c}x_1x_2\\y_1x_2\\z_1z_2\\ \hline y_1z_2\end{array} \Rightarrow 
						       \begin{array}{c}x_1x_2\\y_1x_2\\y_1z_2\\ \hline x_1z_2\end{array}$ \\
    \end{tabular} &
    \begin{tabular}{ c | c | l }
		$p$ & $p'$ & \\
		\hline
		$y_1y_2$ & $x_1x_2$ & \byProposition \\
		& & \\
		$y_1z_2$ & $x_1x_2$ & \byProposition \\
		& & \\
		$z_1x_2$ & $x_1z_2$ & $\begin{array}{c}x_1x_2\\z_1x_2\\z_1z_2\\ \hline x_1z_2\end{array}$ \\
		& & \\
		$z_1y_2$ & $x_1x_2$ & \byProposition \\
		& & \\
		$z_1z_2$ & $x_1x_2$ & \byProposition\\
		& & \\
		& & \\
		& & \\
		& & \\
		
	\end{tabular}
	\end{tabular}
	\\
	\\
	\\
	\\
	Notation $\begin{array}{c}x_1x_2\\y_1y_2\\z_1z_2\\ \hline t_1t_2\end{array}$ should be read as follows
	$$\left.\begin{array}{l}
		f(x_1, x_2) = f(y_1, y_2) = f(z_1, z_2) = 1, \\ 
		p(x_1, y_1, z_1) = t_1, \\ 
		p(x_2, y_2, z_2) = t_2 
	  \end{array}\right\}\Rightarrow f(t_1, t_2) = 1.$$
\end{proof}

\begin{theorem}
    Let some second order conservative language $\Gamma$ have a Maltsev polymorphism $p:X \times X \times X \rightarrow X$. 
    Then $\Gamma$ has a majority polymorphism $p'$, which is the derivative operator of $p$.
\end{theorem}
\begin{proof}
    Indeed, the derivative operator $p'$ is a majority operator according to Lemma \ref{derivativeIsMajorityLemma}.
	The derivative operator $p'$ is conservative and therefore is a polymorphism of all unary functions from $\Gamma$ according to Lemma \ref{conservativeIsPolymorphismLemma}.
	
	Since $\Gamma$ is a conservative language and $p$ is a polymorphism of $\Gamma$, then according to Lemma \ref{conservativePlymorphismLemma}, $p$ is a conservative operator.
	Since $p$ is a conservative Maltsev polymorphism of $\Gamma$, then its derivative operator $p'$ is a polymorphism of all binary functions of $\Gamma$ according to Lemma \ref{derivativeIsPolymorphismLemma}.
	
	Therefore, the derivative operator $p'$ is a majority operator of $\Gamma$.
\end{proof}

\end{document}